%% file: main.tex
\RequirePackage[T1]{fontenc}
\documentclass[journal]{IEEEtran}
\usepackage{comment}
\usepackage{iftex}
\ifPDFTeX
\usepackage[utf8]{inputenc}
\fi
\usepackage{amsmath,amssymb,amsthm,bm}
\usepackage{booktabs,array}
\usepackage{cite}
\usepackage{tikz}
\usetikzlibrary{arrows.meta,positioning,calc}
\usepackage[hidelinks]{hyperref}
\newtheorem{theorem}{Theorem}
\newtheorem{proposition}{Proposition}

\newtheorem{corollary}{Corollary}
\theoremstyle{remark}

\newcommand{\R}{\mathbb R}
\newcommand{\Hh}{\mathcal H}
\newcommand{\ip}[2]{\langle #1,#2\rangle}
\newcommand{\norm}[1]{\lVert #1\rVert}
\newcommand{\Span}{\operatorname{span}}
\newcommand{\rank}{\operatorname{rank}}
\newcommand{\QD}{Q_{\!D}}
\newcommand{\QB}{Q_{\!B}}
\newcommand{\dd}{\mathrm d}
\newcommand{\jj}{\mathrm j}
\newcommand{\avg}[1]{\langle #1\rangle_T}
\begin{document}
\title{Reactive Power Beyond Sinusoids: Scalar Measurement Readouts and Information Limits}
\author{Kai Sun, ~\IEEEmembership{Fellow,~IEEE}
\thanks{K. Sun is with the Min H. Kao Department of Electrical Engineering and Computer Science, The University of Tennessee, Knoxville, TN 37996 USA (e-mail: kaisun@utk.edu).}}
\maketitle

\begin{abstract}
Reactive power is not uniquely defined for nonsinusoidal voltage and current waveforms because different scalar readouts can assign different values to the same measured voltage--current waveform pair. This paper characterizes which scalar readouts are selected by specific measurement and physical requirements, and quantifies the current information lost through scalarization. The scalar readout is first restricted by requiring bilinearity, alternation, and invariance to the measurement time origin. Requiring it to vanish for all smooth, strictly monotone, time-invariant memoryless resistors then forces the harmonic weights to be proportional to harmonic order. With consistency across harmonic cutoffs and fundamental-frequency calibration, this uniquely yields the classical derivative, or Iliovici-type, scalar within the stated class. The selected weighting is also incompatible with conventional reactive-power calibration independently at every harmonic. Furthermore, for known nonzero voltage, active power and a finite set of skew scalar readouts leave an affine family of compatible currents of exact dimension \(d-1-r_v\). These results establish explicit selection and information-loss limits for scalar reactive-power measurements under nonsinusoidal conditions.
\end{abstract}
\begin{IEEEkeywords}
Harmonic analysis, identifiability, nonsinusoidal waveforms, power measurement, power quality, reactive power.
\end{IEEEkeywords}

\section{Introduction}
\IEEEPARstart{R}{eactive} power is an established quantity in electrical power measurement, with roles in metering, equipment loading, compensation, and power-quality characterization. While its conventional definition is well established for single-frequency sinusoidal steady state, nonsinusoidal voltage--current measurements admit different scalar reactive-power readouts for the same waveform pair. Single-frequency complex power combines an average transfer rate with an oriented quadrature quantity. Its success can obscure a special geometric fact: the voltage and current lie in one real two-dimensional harmonic plane. After its orientation is fixed, that plane has only one independent signed area. Distorted periodic waveforms instead occupy several harmonic planes. Reducing their voltage--current relationship to one signed number requires a choice about which planes to weight and which relationships to ignore.

The starting information here is the measured periodic voltage and current, represented by complete Fourier coefficients or harmonic phasors in the declared finite waveform space. Sampled records are equivalent starting data when they suffice to determine those coefficients. A scalar power readout is a measurement functional applied to this information; retaining only selected readouts is a further information reduction. The paper studies the deterministic mathematical properties of this waveform-to-scalar reduction, assuming the input coefficients are available. It contributes to the theory and interpretation of electrical power measurements; acquisition hardware and stochastic measurement uncertainty are outside its scope.

Existing approaches to nonsinusoidal reactive or non-active power can be broadly divided into two complementary directions: scalar power definitions and current decomposition. Scalar approaches seek a single quantity that extends conventional reactive power. Budeanu’s definition sums the reactive powers of individual harmonics with equal weighting \cite{Czarnecki1987,Willems2011}, whereas Iliovici-type loop integrals and differentiated-power formulations introduce frequency-dependent weighting through time differentiation \cite{Jeltsema2014,Jeltsema2003,GarciaCanseco2005}. Current-decomposition approaches retain more information about the waveform. Fryze separates the current into active and non-active components \cite{CzarneckiSwietlicki1990,Willems2011}, while the currents’ physical components (CPC) framework \cite{Czarnecki2008} further decomposes the current according to different physical effects and compensation objectives.

These questions have also been studied from an instrumentation-and-measurement perspective. Prior work examined the interpretation of Budeanu reactive power and related current components \cite{Czarnecki1987,Willems2011}, while Czarnecki and Swietlicki \cite{CzarneckiSwietlicki1990} explicitly connected nonsinusoidal-power interpretation, analysis, and measurement. IEEE Std.~1459-2025 further provides standardized definitions for the measurement of electric-power quantities under sinusoidal and nonsinusoidal conditions \cite{IEEE1459_2025}. Building on this measurement literature and the CPC distinction between scalar power quantities and richer current descriptions, this paper addresses a complementary question: given measured voltage and current waveforms, what structural and physical requirements select a scalar reactive-power readout, and what current information is lost through scalarization? The objective is therefore to characterize the selection and information content of scalar measurement functionals, rather than to propose a replacement for standardized metering quantities. 

Vector, exterior, and geometric-algebra power theories make the underlying multidimensional relationship explicit. In particular, LaWhite and Ili'{c} introduced an antisymmetric exterior power matrix, its norm, conservation identities, and scalar projections \cite{LaWhiteIlic1997}. Their discussion, together with Willems's comparison of sinusoidal and nonsinusoidal nonactive-current spaces \cite{Willems2011}, already identifies the dimensional difficulty of representing such information by a single scalar. Later geometric formulations develop related decompositions and applications \cite{Menti2007,Montoya2020,Montoya2021,Cieslinski2024}. In this paper, the exterior-algebra representation is used only as an alternative organization of the voltage--current information. When the nonzero voltage waveform is known, complete real Fourier coefficients, complex harmonic phasors, and the compatible pair $(P,v\wedge i)$ contain equivalent information about the current waveform. Rewriting the same data in a different mathematical form therefore does not provide additional measurements or recover information that was not measured. The measurement chain considered here starts from measured voltage and current waveforms, represents them through waveform coefficients or phasors, and then reduces that information to scalar power readouts.

The first representation retains the measured waveform information, whereas the final scalarization may discard part of it. This paper focuses on this last mapping: how a scalar readout is selected and what information is retained or lost once the waveform measurements are reduced to scalars.

Two measurement questions follow. First, \emph{measurement-function selection}: which structural and physical requirements select a scalar reactive-power readout from measured waveforms? Common time-shift invariance removes dependence on the arbitrary measurement time origin, but within the real-bilinear, alternating class the harmonic weights remain free; memoryless-resistor nulling provides an additional physical constraint. Agreement on sinusoidal tests alone does not determine a readout for distorted waveforms. Second, \emph{information loss caused by scalarization}: for known nonzero voltage, what current information remains identifiable from active power and finite skew scalar readouts? For compatible data, Theorem~\ref{thm:rank} of the paper gives the remaining ambiguity dimension \(d-1-r_v\), where \(d\) is the waveform-space dimension and \(r_v\) the rank of the observed skew directions. Even the full common-shift-invariant alternating class leaves in-phase redistribution directions when multiple positive-frequency voltage planes are excited, so identical scalar power records can correspond to different RMS currents, harmonic conductance distributions, conductor-loss burdens, and compensation requirements. This characterizes measurement sufficiency for power metering and power-quality analysis within the stated waveform model, complementing standardized quantities such as IEEE Std.~1459-2025 \cite{IEEE1459_2025}.

The primary contribution of the paper is Theorem~\ref{thm:selection}, a converse characterization of scalar measurement readouts on contiguous finite harmonic spaces: within the real-bilinear, alternating, common-time-shift-invariant class, universal nulling for the stated smooth, strictly monotone, time-invariant memoryless resistors forces harmonic-order weighting $a_n=cn$. Fixed-reference compatibility and fundamental-frequency calibration then give $c=1$. The derivative/Iliovici scalar and its memoryless-resistor cancellation are established results. Prior studies of generalized content and differentiated power address nonlinear circuit balance and related reactive quantities \cite{SupertiFurgaPinola1994,SupertiFurga1994,Jeltsema2003,GarciaCanseco2005}, while the Lissajous formulation of Hong and de Le'on provides a physical circuit interpretation \cite{HongDeLeon2015}. Recent work \cite{KosobudzkiLadniak2026} further advocates derivative weighting as a unique additive extension for periodic waveforms and uses memoryless-resistor nulling as a selection criterion among derivative/integral candidates. In contrast, Theorem~\ref{thm:selection} starts from the full declared bilinear, alternating, common-time-shift-invariant class and derives $a_{m+n}=a_m+a_n$ from admissible nonlinear mixing, thereby forcing harmonic-order weighting up to calibration. Thus, no novelty or historical priority is claimed for the derivative/Iliovici formula or resistor nulling itself. Table~\ref{tab:prior} summarizes the relation to established approaches.

\begin{table}[t]
\caption{Established results and their role in this paper}
\label{tab:prior}
\centering\footnotesize
\renewcommand{\arraystretch}{1.2}
\setlength{\tabcolsep}{1.3pt}
\begin{tabular}{@{}>{\raggedright\arraybackslash}p{.28\columnwidth}
                    >{\raggedright\arraybackslash}p{.38\columnwidth}
                    >{\raggedright\arraybackslash}p{.33\columnwidth}@{}}
\toprule
Prior topic & Established result & Role here \\
\midrule

Iliovici/differentiated power \cite{Jeltsema2014}
& Derivative scalar; resistor cancellation
& Basis for the selection problem \\

Recent additive extension \cite{KosobudzkiLadniak2026}
& Derivative weighting selected among derivative/ integral candidates
& Converse over the full declared invariant bilinear class \\

Generalized content \cite{SupertiFurgaPinola1994,SupertiFurga1994}
& Nonlinear reactive quantities and network balance
& Scope-specific scalar characterization \\

Exterior/vector power \cite{LaWhiteIlic1997}
& Multidimensional power data and scalar projections
& Scalar selection and information limits \\

CPC/scattered current \cite{Czarnecki2008,Czarnecki2019}
& In-phase redistribution and compensation
& Kernel of the scalar observations \\

Geometric algebra \cite{Menti2007,Montoya2021,Cieslinski2024}
& Related representations and current decompositions
& Equivalent terminal information; retained directions \\

\bottomrule
\end{tabular}
\end{table}

Beyond the main scalar-selection theorem, the paper makes three additional contributions:
\begin{enumerate}
\item It characterizes all common-time-shift-invariant scalar readouts as weighted harmonic quadrature sums and shows that resistor nulling conflicts with equal calibration across all harmonic planes.
\item It quantifies scalarization loss: for known nonzero voltage, active power and finite skew scalar readouts leave a compatible-current family of dimension $d-1-r_v$, with uniqueness exactly when this dimension is zero.
\item It identifies the residual ambiguity of the full invariant scalar class: with $N$ excited positive-frequency harmonic planes and no DC, exactly $N-1$ in-phase redistribution directions remain unobserved. The examples show consequences for RMS current, loss burden, compensation, and reactive-power interpretation.
\end{enumerate}
Reconstruction, Gram identities, and rank--nullity are used only as tools; the contribution is their application to scalar reactive-power measurements and their information limits.

\section{Waveform Space and the Prior-Art Baseline}
\label{sec:framework}
\subsection{Conventions and harmonic planes}
Fix a reference $\omega_0>0$, period $T=2\pi/\omega_0$, passive terminal current direction, and define the real RMS inner product for periodic signals $x$ and $y$ (such as $v$ and $i$) by
\begin{equation}
 \ip{x}{y}=\frac1T\int_0^T x(t)y(t)\,\dd t,
 \qquad
\|x\|:=\sqrt{\langle x,x\rangle}.
\label{eq:inner}
\end{equation}
The RMS voltage and current are therefore $V=\|v\|,I=\|i\|$.
Unless otherwise stated, coefficients multiplying unnormalized trigonometric functions denote peak amplitudes in this paper.

The finite spaces for the selection theorem are
\begin{equation}
 \Hh_N^0=\bigoplus_{n=1}^N H_n,
 \qquad \Hh_N^+=\R1\oplus\Hh_N^0,
 \label{eq:spaces}
\end{equation}
where each $H_n=\Span\{e_{n,c},e_{n,s}\}$ is a two-dimensional real vector space of functions with $e_{n,c}:=\sqrt2\cos(n\omega_0t)$ and $e_{n,s}:=\sqrt2\sin(n\omega_0t)$. 
A superscript $0$ with $\Hh_N$ excludes DC; a superscript $+$ includes it. Simply write $\Hh_N$ when either choice is allowed, and let \(\Pi_N\) denote the orthogonal projection of the ambient waveform space, finite- or infinite-dimensional, onto the selected harmonic subspace \(\mathcal H_N\). The scalar channel assumption is important: repeated harmonic representations in a multichannel space admit additional invariant operators.

On each harmonic subspace $H_n$, write
\begin{equation}
v_n=v_{n,c}e_{n,c}+v_{n,s}e_{n,s}, \qquad
i_n=i_{n,c}e_{n,c}+i_{n,s}e_{n,s}.
\label{eq:vn_in}
\end{equation}
Define the linear quadrature operator $J_n:H_n\to H_n$ by $J_ne_{n,c}=e_{n,s},\qquad
J_ne_{n,s}=-e_{n,c}.$
Thus, $J_n^2=-I$ and, under the adopted sinusoidal convention,
$J_n$ represents a $-90^\circ$ phase shift on $H_n$.
For $D=\dd/\dd t$, we have $D|_{H_n}=-n\omega_0J_n.$

For the exterior product on $H_n$, define the oriented unit bivector
$E_n:=e_{n,c}\wedge e_{n,s}$, which spans the one-dimensional bivector
space $\Lambda^2 H_n=\operatorname{span}\{E_n\}$.
By bilinearity and antisymmetry, $e_{n,c}\wedge e_{n,c}
=e_{n,s}\wedge e_{n,s}=0,\quad
e_{n,s}\wedge e_{n,c}=-E_n.$
With the above orientation and quadrature convention, 
\[
\begin{aligned}
v_n\wedge i_n
&=
(v_{n,c}e_{n,c}+v_{n,s}e_{n,s})
\wedge
(i_{n,c}e_{n,c}+i_{n,s}e_{n,s}) \\
&=
(v_{n,c}i_{n,s}-v_{n,s}i_{n,c})E_n.
\end{aligned}
\]
Define the scalar coefficient
\begin{equation}
Q_n:=v_{n,c}i_{n,s}-v_{n,s}i_{n,c}.
\label{eq:Qn}
\end{equation}
Equivalently, using the operator $J_n$, we have $
Q_n=\langle J_nv_n,i_n\rangle.$

Because the waveform coordinates multiply
$\sqrt{2}\cos(n\omega_0t)$ and $\sqrt{2}\sin(n\omega_0t)$, respectively, in \eqref{eq:vn_in}, note that the conventional electrical-engineering phasor convention corresponds to a minus sign on the sine coordinate:
\[
\underline V_n=v_{n,c}-\jj v_{n,s},\qquad
\underline I_n=i_{n,c}-\jj i_{n,s}.
\]
The same scalar coefficient satisfies
\[
Q_n=\operatorname{Im}
\!\left(\underline V_n\underline I_n^\ast\right).
\]
Hence, a lagging inductive current corresponds to $Q_n>0$.

\subsection{Complete data before scalarization}
Let $H$ now be any finite real inner-product space of dimension $d$. Define
\begin{equation}
 P=\ip v i,\qquad B=v\wedge i.
 \label{eq:baseline}
\end{equation}
Here $v=(v_1,\ldots,v_d)^{\mathsf T}$ and $i=(i_1,\ldots,i_d)^{\mathsf T}$ denote coordinates in an orthonormal basis of $H$. For conventional matrix calculations, the bivector $B$ may equivalently be represented by a skew-symmetric matrix
\begin{equation}
K:=vi^{\mathsf T}-iv^{\mathsf T},
\label{eq:K}
\end{equation}
whose independent entries $K_{rs}=v_r i_s-i_r v_s
=-K_{sr}$ are $B$'s coordinates. Thus $B$ and $K$ contain the same geometric information. The induced norm on exterior space $\Lambda^2 H$ obeys
\begin{equation}
\norm{B}^2=\tfrac12\norm{K}_F^2=V^2I^2-P^2, \quad
\norm{B}=V\norm{i-\tfrac{P}{V^2}v},
\label{eq:gram}
\end{equation}
where the last identity assumes $v\ne0$ and $\norm K_F= \sqrt{\operatorname{tr}(K^{\mathsf T}K)}$
denotes the Frobenius norm. These are the Gram identity and the active-current projection, not new physical power laws \cite{LaWhiteIlic1997,Willems2011}. The full wedge generally contains cross-frequency coordinates, such as
$(v_{m,c}i_{n,s}-v_{n,s}i_{m,c})\,e_{m,c}\wedge e_{n,s}$ for $m\neq n$,
in addition to the diagonal coordinates $Q_nE_n$.

Let $\lrcorner$ denote the left contraction, defined by \(x\mathbin{\lrcorner}(a\wedge b)
=\langle x,a\rangle b-\langle x,b\rangle a,
\) and extended linearly to general bivectors. The following proposition establishes a reconstruction property of the compatible \((P,B)\) representation.

\begin{proposition}[Reconstruction and Compatibility]
\label{prop:reconstruction}
For fixed $v\ne0$, supplied $P\in\R$ and $B\in\Lambda^2H$ are compatible with a current if and only if $B\in v\wedge H$, equivalently $v\wedge B=0$. The unique current is
\begin{equation}
 i=\frac{Pv+v\mathbin{\lrcorner}B}{V^2}
   =\frac{Pv-Kv}{V^2}.
 \label{eq:reconstruct}
\end{equation}
\end{proposition}
\begin{proof}
For the forward direction (\emph{only if}), suppose a current $i$ exists
with $P=\langle v,i\rangle$ and $B=v\wedge i$. Then
\(v\wedge B=v\wedge(v\wedge i)=0,\)
so $B$ satisfies the compatibility condition. Moreover,\(
v\mathbin{\lrcorner}B
=v\mathbin{\lrcorner}(v\wedge i)
=V^2i-Pv=-Kv,\)
which gives \eqref{eq:reconstruct} and establishes uniqueness.
For the converse (\emph{if}), choose the unit vector $e_1=v/V$
pointing along $v$ and decompose
\(
B=e_1\wedge b+C,
\)
with $b\in e_1^\perp$ and $C\in\Lambda^2(e_1^\perp)$.
For example, in three dimensions, if
$e_1^\perp=\operatorname{span}\{e_2,e_3\}$, then
$b=b_2e_2+b_3e_3$ and
$C=c\,e_2\wedge e_3$.
Since
\(
v\wedge B=V e_1\wedge C,
\)
the condition $v\wedge B=0$ implies $C=0$.
Substitution in \eqref{eq:reconstruct} then gives
$P=\langle v,i\rangle$ and $B=v\wedge i$.
\end{proof}

This elementary result establishes the comparison baseline: \emph{full compatible $(P,B)$ and known nonzero $v$ preserve $i$}. It does not identify a circuit realization or its energy. An arbitrary ambient bivector is not a compatible perturbation at a fixed voltage. In particular, if $\dim H=d$, counting all $\binom{d}{2}$ exterior
coordinates would overstate the independent current information, since
for fixed $v\ne0$ the compatible bivectors form the $(d-1)$-dimensional
space $v\wedge H$, corresponding to the $d-1$ current components
orthogonal to $v$.

\subsection{Why a signed scalar needs structure}
Every real-bilinear alternating scalar form $q_A$ on finite-dimensional $H$ admits a unique representation by a skew-adjoint operator $A$:
\begin{equation}
 q_A(v,i)=\ip{Av}{i}=-\ip{v}{Ai},\qquad A^*=-A.
 \label{eq:skew}
\end{equation}
Here \(A^\ast\) denotes the adjoint of \(A\) (in a real orthonormal coordinate basis, $A^*=A^{\mathsf T}$). 

Equivalently, \(q_A\) is a linear scalar readout of \(B=v\wedge i\), obtained by taking its inner product with the bivector associated with the skew-adjoint operator \(A\). This is a finite-dimensional representation result: the skew-adjoint operator \(A\) specifies the readout rule and should not be interpreted as the load admittance, even though it acts on \(v\). For example, consider the bivector
\(
B
=B_{12}\,e_1\wedge e_2
+B_{13}\,e_1\wedge e_3
+B_{23}\,e_2\wedge e_3
\) in a three-dimensional space.
A particular skew-adjoint operator $A$ may correspond to the bivector
direction
\(
c_{12}\,e_1\wedge e_2
+c_{13}\,e_1\wedge e_3
+c_{23}\,e_2\wedge e_3,
\)
so that the resulting scalar readout is
\(
q_A(v,i)
=c_{12}B_{12}+c_{13}B_{13}+c_{23}B_{23}.
\)
Different choices of $A$ therefore produce different scalar projections of the same $B$, each retaining only part of its geometric
information.

The following result explains why a nonzero scalar describing the antisymmetric part of the voltage--current relationship cannot be defined in dimensions three or higher without introducing additional structure.

\begin{proposition}[Structure-free scalar obstruction]
\label{prop:nogo}
If $d\ge 3$, an alternating bilinear scalar invariant under all
orientation-preserving orthogonal transformations of $H$ is identically zero.
\end{proposition}
\begin{proof}
Suppose, to the contrary, that such a scalar $q$ exists and is not
identically zero. Then there exist two orthonormal vectors
$e_1,e_2\in H$ such that $q(e_1,e_2)\neq 0$.
Since $d\ge 3$, choose a third unit vector $e_3$ orthogonal to both.
Define an orthogonal transformation $R$ by
$R e_1=-e_1$, $R e_2=e_2$, $R e_3=-e_3$, with all remaining
orthogonal directions unchanged. Because two directions are reversed,
$\det R=1$, so $R$ is orientation preserving. By the assumed invariance of
$q$,
\(
q(e_1,e_2)=q(Re_1,Re_2).
\)
Using bilinearity,
\(
q(Re_1,Re_2)=q(-e_1,e_2)=-q(e_1,e_2).
\)
Therefore,
\(
q(e_1,e_2)=-q(e_1,e_2),
\)
which implies $q(e_1,e_2)=0$. This contradicts the assumption that
$q(e_1,e_2)\neq 0$. Hence no nonzero scalar with the stated properties
exists, and therefore $q_A\equiv 0$.
\end{proof}

On an oriented two-dimensional plane, the area scalar is rotation-invariant and determines the wedge, while reflection reverses its sign. In higher dimensions, any nonzero scalar requires additional structure. Proposition~\ref{prop:nogo} therefore serves only as a structure-free baseline: it does not exclude norms, model-specific quantities, or temporally selected scalars, nor does it treat arbitrary rotations as physical circuit operations. Under coordinate changes, a specified physical readout is transported with its operator. The physically relevant selection result below instead uses common time shifts.

\section{Temporal Symmetry and Admissible Scalar Readouts}
\label{sec:shifts}
Common time shifts act as $U_\tau x(t)=x(t+\tau)$ for all real $\tau$, modulo $T$. On each harmonic plane, this operation is a rotation at the corresponding integer multiple of the reference angular frequency. The same shift is applied to voltage and current, representing a change in the measurement time origin. Time-shift invariance does not select a unique reactive-power scalar, but it restricts all admissible alternating bilinear scalars to weighted sums of the conventional harmonic reactive powers $Q_n$, as stated in the following proposition.

\begin{proposition}[Time-shift classification]
\label{prop:shifts}
On either space \eqref{eq:spaces}, an alternating bilinear scalar is invariant under every common time shift if and only if
\begin{equation}
 A|_{H_n}=a_nJ_n,\qquad
 q_A(v,i)=\sum_{n=1}^N a_nQ_n,
 \quad a_n\in\R.
 \label{eq:weights}
\end{equation}
If DC is included and $1(t)$ denotes the constant waveform, then the operator $A$ maps the DC component to zero, i.e. $A1(t)=0$, and
no DC--harmonic coupling terms occur.
\end{proposition}

\begin{proof}
Let $q_A(v,i)=\langle Av,i\rangle$, where $A$ is skew-adjoint. Time-shift invariance requires
$q_A(U_\tau v,U_\tau i)=q_A(v,i)$ for every $\tau$. Since $U_\tau$ is orthogonal, this is equivalent to $AU_\tau=U_\tau A$.

On each harmonic plane $H_n$, $U_\tau=\exp(-n\omega_0\tau J_n)$, a rotation of the coefficient vector by $-n\omega_0\tau$. Because different harmonic planes rotate at different rates, any operator commuting with every $U_\tau$ must preserve each $H_n$ and cannot contain cross-harmonic couplings. On each $H_n$, a real linear operator commuting with all rotations has the form
$A|_{H_n}=g_nI+a_nJ_n$. Since $A$ is skew-adjoint while $I$ is self-adjoint, we must have $g_n=0$, and hence $A|_{H_n}=a_nJ_n$. Therefore
$q_A(v,i)=\sum_{n=1}^N a_n\langle J_nv_n,i_n\rangle=\sum_{n=1}^N a_nQ_n$.

If DC is included, a skew-adjoint operator on this one-dimensional DC subspace is zero, so $A(1)=0$, and no DC--harmonic couplings occur. Conversely, operators with blocks $a_nJ_n$ and no cross-couplings commute with every common time shift, proving the result.
\end{proof}

\emph{Remark:} This proposition is important because it reduces the search for a time-shift-invariant scalar reactive-power measure to the choice of harmonic weights $a_n$. In other words, once bilinearity, alternation, and common time-shift invariance are imposed, the remaining freedom is entirely in how the individual harmonic reactive powers $Q_n$ are weighted.

Several familiar choices fit \eqref{eq:weights}. Equal weights give $\QB=\sum_{n=1}^N Q_n$, derivative weights give $\QD=\sum_{n=1}^N nQ_n$, and the inverse-frequency readout $e_\Delta:=\sum_{n=1}^N Q_n/(2n\omega_0)$ has energy units. Here $\omega_0$ is a fixed reference angular frequency and $N$ is the retained positive-harmonic cutoff. Table~\ref{tab:selectors} summarizes these choices without identifying any of them with generic stored energy. The analysis assumes this fixed reference and finite harmonic bandwidth; redefining $\omega_0$ for individual harmonics changes the functional being compared, while frequency-estimation error, spectral leakage, and measurement noise are outside the present deterministic scope.

No linearity of $i$ as a function of $v$ is assumed. Cross-frequency wedge coordinates may be nonzero in either linear or nonlinear circuits, but they cannot contribute to a scalar within this invariant bilinear class. They are excluded by the distinct temporal behavior of different harmonics, not because they are physically unimportant. Scalars using an external phase reference or a time-varying measurement operator belong to a different observation class.

\begin{table}[t]
\caption{Selected readouts on a fixed harmonic grid}
\label{tab:selectors}
\centering\footnotesize
\renewcommand{\arraystretch}{1}
\setlength{\tabcolsep}{5pt}
\begin{tabular}{@{}>{\raggedright\arraybackslash}p{.1\columnwidth}>{\raggedright\arraybackslash}p{.13\columnwidth}>{\raggedright\arraybackslash}p{.7\columnwidth}@{}}
\toprule
Readout & Weight(s) & Structure and interpretation \\
\midrule
$\QB$
& $1$
& Equal harmonic weighting; preserves single-harmonic calibration. No universal nonlinear-resistor null. \\

$\QD$
& $n$
& Harmonic-order weighting; vanishes for the stated time-invariant memoryless resistors. Not a general stored-energy measure. \\

$e_\Delta$
& $\frac{1}{2n\omega_0}$
& Energy units; represents storage imbalance under the stated LTI assumptions. No universal resistor null. \\

All $Q_n$
& Separate terms
& Retains each harmonic quadrature term, but misses $N-1$ in-phase directions when all $N$ voltage planes are excited (no DC). \\
\bottomrule
\end{tabular}
\end{table}

\section{Selection by Nonlinear-Resistor Nulling}
\label{sec:selection}
The time-shift classification reduces scalar selection to the harmonic weights $a_n$ but does not determine them. To constrain these weights, this section uses time-invariant memoryless nonlinear resistors, which can generate harmonics and nontrivial voltage--current geometry without electrical storage. Requiring a storage- or circulation-oriented scalar to vanish on such devices is used here only as a selection criterion; it does not imply absence of nonactive current, harmonic redistribution, or nonscalar voltage--current geometry. Theorem~\ref{thm:selection} therefore asks the narrower converse question: \emph{under bilinearity, alternation, common time-shift invariance, and universal resistor nulling, which scalar weighting is possible?}

A storage-based physical viewpoint has also been advocated for nonlinear circuits, including the Lissajous-based formulation of Hong and de Le\'on \cite{HongDeLeon2015}. The present work does not adopt a particular equivalent-circuit reconstruction; instead, it asks what a memoryless-resistor nulling requirement implies for a declared class of terminal bilinear scalar readouts.

\subsection{The precise nulling axiom}
A nonlinear constitutive law does not preserve a finite harmonic space. Even when instantaneous terminal voltage $v\in\mathcal H_N$, the resulting current $i=f(v)$ may contain DC and frequencies above the voltage cutoff. The finite nulling axiom must therefore be stated as
\begin{equation}
 \ip{A_Nv}{\Pi_N f(v)}=0
 \quad\text{for every }v\in\Hh_N\text{ and }f\in\mathcal F,
 \label{eq:null}
\end{equation}
where $\mathcal F$ consists of smooth, globally strictly increasing, time-invariant, single-valued real functions with $f(0)=0$. $vf(v)>0$ for $v\ne0$, implying strictly positive instantaneous dissipation. Time invariance and single-valuedness exclude explicit time modulation, hysteresis, and internal memory. Projection in \eqref{eq:null} is only an observation operation. Since $A_Nv\in\Hh_N$, its pairing with $\Pi_Nf(v)$ equals its pairing with the complete current:
\begin{equation}
\langle A_Nv,\Pi_N f(v)\rangle
=
\langle A_Nv,f(v)\rangle.
\label{eq:projection_identity}
\end{equation}
Thus, $\Pi_N$ represents only the finite-dimensional observation used by the scalar readout; it does not replace the physical constitutive law $i=f(v)$ by a truncated one.

\begin{theorem}[Restricted derivative selection]
\label{thm:selection}
Let $N\ge1$ and let $q_{A_N}$ on $\Hh_N^0$ or $\Hh_N^+$ be real-bilinear, alternating, and invariant under all common shifts. It satisfies \eqref{eq:null} if and only if there exists a real constant $c_N$
such that
\begin{equation}
 a_n=c_Nn\ (1\le n\le N),\qquad
 A_N=-\frac{c_N}{\omega_0}D|_{\Hh_N}.
 \label{eq:selected}
\end{equation}
If the forms agree under the standard inclusions $\Hh_N\subset\Hh_{N+1}$ with fixed DC convention and reference, then $c_N=c$. Fundamental calibration $q_{A_N}(e_{1,c},e_{1,s})=1$ fixes $c=1$.
\end{theorem}
\begin{proof}
By Proposition~\ref{prop:shifts}, common time-shift invariance already implies $A_N|_{H_n}=a_nJ_n, \quad
q_{A_N}(v,i)=\sum_{n=1}^N a_nQ_n.$ 
It therefore remains to determine which harmonic weights $a_n$ are compatible
with the resistor-nulling condition \eqref{eq:null}.

First, it will be shown that this condition imposes a quadratic mixing identity. For any
fixed $v\in\Hh_N$, choose $M>\norm v_\infty$ and a smooth compactly supported
function $\chi$ that equals one on $[-M,M]$. Define
\[
h(x)=\chi(x)x^2,\qquad
f(x)=Gx+\epsilon h(x),
\]
where $\epsilon\neq0$ introduces the quadratic mixing term, $\chi$ leaves it
unchanged over the voltage range of interest while cutting it off outside a
bounded interval, and $G$ ensures global monotonicity:
\begin{equation}
G>|\epsilon|\sup_x|h'(x)|.
\label{eq:passive-extension}
\end{equation}
Then $f'(x)>0$ for all $x$, so $f$ belongs to the admissible resistor class,
while over the range actually reached by $v(t)$,
\[
f(v)=Gv+\epsilon v^2.
\]
Such a test law may be chosen separately for each fixed waveform $v$, since
\eqref{eq:null} is required to hold for every admissible pair $(v,f)$.

Applying \eqref{eq:null}, using the projection identity \eqref{eq:projection_identity}, gives
\[
0=\ip{A_Nv}{f(v)}
  =G\ip{A_Nv}{v}
   +\epsilon\ip{A_Nv}{v^2}.
\]
Because $A_N$ is skew-adjoint, $\ip{A_Nv}{v}=0$. Since $\epsilon\ne0$,
\begin{equation}
\ip{A_Nv}{v^2}=0
\qquad\text{for every }v\in\Hh_N.
\label{eq:quadratic}
\end{equation}

Next, use \eqref{eq:quadratic} to relate the weights of different harmonics.
Substitute $v=rx+sy+tz$ into \eqref{eq:quadratic} and collect the coefficient
of the product $rst$. After removing a common factor, this gives
\begin{equation}
\ip{A_Nx}{yz}
+\ip{A_Ny}{xz}
+\ip{A_Nz}{xy}=0.
\label{eq:polarization}
\end{equation}
Now let $x=\cos m\theta$, $y=\cos n\theta$, $z=\sin(m+n)\theta$, $\theta=\omega_0t$,
with $m,n\ge1$ and $m+n\le N$. Since
$A_N|_{H_k}=a_kJ_k$, 
\[
\ip{A_Nx}{yz}
=
a_m\ip{\sin m\theta}{\cos n\theta\,\sin(m+n)\theta}.
\]
Using elementary product-to-sum identities and orthogonality of distinct harmonics gives
\[
\ip{A_Nx}{yz}=\frac{a_m}{4},
\]
and the other two terms follow analogously:
\[
\ip{A_Ny}{xz}=\frac{a_n}{4},
\qquad
\ip{A_Nz}{xy}=-\frac{a_{m+n}}{4}.
\]
Equation~\eqref{eq:polarization} therefore requires
\begin{equation}
a_{m+n}=a_m+a_n,
\qquad m+n\le N.
\label{eq:additivity}
\end{equation}
Taking $m=1$ successively gives
\[
a_2=2a_1,\quad a_3=3a_1,\quad\ldots,\quad a_n=na_1.
\]
Hence, with $c_N=a_1$,
\[
a_n=c_Nn.
\]
Since $D|_{H_n}=-n\omega_0J_n$, this is equivalently
\[
A_N=-\frac{c_N}{\omega_0}D|_{\Hh_N},
\]
proving the necessity of \eqref{eq:selected}. For $N=1$, the same
expression holds directly with the unrestricted constant $c_N=a_1$.

For sufficiency, suppose \eqref{eq:selected} holds. Since $f$ is smooth,
let $F$ be a primitive of $f$, i.e., $F'=f$. Using the memoryless
relation $i=f(v)$ and periodicity of $v(t)$, 
\begin{align}
\ip{A_Nv}{\Pi_Nf(v)}
&=\ip{A_Nv}{f(v)} \nonumber\\
&=-\frac{c_N}{\omega_0T}
  \int_0^T f(v)\dot v\,\dd t \nonumber\\
&=-\frac{c_N}{\omega_0T}
  [F(v)]_0^T=0.
\label{eq:primitive}
\end{align}
Thus the selected weights satisfy the nulling condition. Apart from the
regularity needed above, the cancellation uses only periodicity and the
time-invariant memoryless relation $i=f(v)$; monotonicity and passivity are
not used.

Finally, compatibility across successive harmonic cutoffs requires
$q_{A_{N+1}}$ restricted to $\Hh_N\times\Hh_N$---i.e. to all $(v,i)$ pairs with $v,i\in\Hh_N$---to agree with $q_{A_N}$.
In particular, their actions on the common fundamental plane $H_1$ must
agree. Hence
\[
A_{N+1}|_{H_1}=A_N|_{H_1},
\]
which implies $c_{N+1}J_1=c_NJ_1$ and therefore $c_{N+1}=c_N$.
Thus all truncations share one constant $c$. The calibration
$q_{A_N}(e_{1,c},e_{1,s})=1$ then gives $c=1$.
\end{proof}

The selected functional is therefore the classical normalized derivative,
or Iliovici-type, scalar \cite{Jeltsema2014}:
\begin{equation}
 \QD(v,i)=\sum_{n=1}^N nQ_n
 =-\frac{\ip{Dv}{i}}{\omega_0}
 =\frac{\ip v{Di}}{\omega_0}.
 \label{eq:derivative}
\end{equation}
The last equality follows from periodic integration by parts.
The contribution of Theorem~\ref{thm:selection} is not the definition of
$\QD$ itself, but the converse result that the stated symmetry and
memoryless-resistor nulling requirements force this harmonic weighting,
up to calibration.

\subsection{Scope and Limitations of Theorem~\ref{thm:selection}}
The theorem is proved on a contiguous positive-frequency harmonic space
\(
\Hh_N^0=H_1\oplus H_2\oplus\cdots\oplus H_N,
\)
optionally augmented by DC. It contains every harmonic order from $1$ through $N$. This assumption is used in the proof because the quadratic mixing argument couples harmonic orders $m$, $n$, and $m+n$, leading to
\(
a_{m+n}=a_m+a_n.
\)
If the retained spectrum is sparse, for example if some intermediate harmonics are absent, such mixing relations may not be available for all pairs $(m,n)$, and the theorem does not by itself imply $a_n\propto n$ on that sparse set. By contrast, the information-loss result in Section~\ref{sec:information} requires only that the retained harmonic orders be distinct and positive.

The stated resistor family is a sufficient test family and is not claimed to be minimal. The necessity proof uses second-order mixing: a smooth law with $f''(0)\ne0$, tested at sufficiently small amplitudes, already yields the quadratic identity by expansion. Additional constitutive symmetries can reduce the available mixing identities. In particular, restricting the family to odd characteristics, $f(-v)=-f(v)$, excludes the quadratic test used here; an odd-only characterization may require a different argument and is not covered by this necessity proof. This does not assert that the conclusion fails for odd laws. Linear resistors impose no weight constraint because $\ip{A_Nv}{v}=0$. Identifying the weakest physical test class that still forces $a_n\propto n$ is outside the present scope.

DC does not need to be included in the test voltage because the witness waveforms used in the proof can be chosen with zero mean, and the operator $A_N$ acts trivially on the DC component. However, a nonlinear memoryless law can still generate a DC component in the current even when the applied voltage has zero mean. This causes no difficulty because the scalar readout depends only on the projection of the current onto the retained harmonic space. In particular, removing the voltage DC component before evaluating $f(v)$ would change the nonlinear input to the resistor and would therefore represent a different physical test. Similarly, imposing zero mean on the full nonlinear current would unnecessarily restrict the admissible resistor class.

Theorem~\ref{thm:selection} is formulated for a finite harmonic space. Extending the derivative-weighted scalar to arbitrarily high harmonics requires additional smoothness or harmonic-decay assumptions, because the weight on $Q_n$ grows linearly with harmonic order. For example, on a single normalized harmonic pair $(e_{n,c},e_{n,s})$, one has $\QD=n$, so the functional is not bounded on unrestricted $L^2_{\mathrm{per}}\times L^2_{\mathrm{per}}$, where $L^2_{\mathrm{per}}$ denotes periodic waveforms with finite RMS value. If the voltage has one square-integrable derivative, $v\in H^1_{\mathrm{per}}$, then
\begin{equation}
 |\QD(v,i)|
 \le \omega_0^{-1}\norm{Dv}\norm i,
 \label{eq:bound}
\end{equation}
which guarantees continuity for $i\in L^2_{\mathrm{per}}$. Thus, increasing the measurement bandwidth is not mathematically neutral for the derivative-weighted scalar: sufficient waveform smoothness or high-frequency decay is needed for the infinite-harmonic limit to remain well behaved.

\section{Incompatibility with Every-Harmonic Calibration}
\label{sec:conflict}
Theorem~\ref{thm:selection} identifies the harmonic weighting selected by the memoryless-resistor nulling criterion. The following result shows the corresponding tradeoff: this weighting cannot simultaneously preserve the conventional single-frequency reactive-power calibration on every harmonic plane.

\begin{corollary}[Calibration conflict]
\label{cor:conflict}
For $N\ge2$, no readout in Theorem~\ref{thm:selection}'s class both satisfies \eqref{eq:null} and agrees with conventional $Q_n$ on every isolated plane $H_n$.
\end{corollary}
\begin{proof}
Testing $(e_{n,c},e_{n,s})$ on each isolated plane requires $a_n=1$. Theorem~\ref{thm:selection}, calibrated at $n=1$, requires $a_n=n$. These already conflict at $n=2$.
\end{proof}

A simple two-harmonic example uses the voltage
\(
v=v_1+v_2=\cos\theta+\sin2\theta,
\)
where both harmonic peak amplitudes are set to unity for convenience.
Construct \eqref{eq:passive-extension} with $M>2$. Over this voltage range,
the current is exactly $i=Gv+\epsilon v^2$, while the globally defined
constitutive law remains smooth, strictly monotone, and passive. Since
\begin{equation}
 v^2=1+\tfrac12\cos2\theta-\tfrac12\cos4\theta
       +\sin\theta+\sin3\theta,
 \label{eq:two-expansion}
\end{equation}
the current components are
$i_0=\epsilon$, $i_1=G\cos\theta+\epsilon\sin\theta$, $i_2=G\sin2\theta+\frac{\epsilon}{2}\cos2\theta$, $i_3=\epsilon\sin3\theta$,
$i_4=-\frac{\epsilon}{2}\cos4\theta$.
If $V_{n,c},V_{n,s}$ and $I_{n,c},I_{n,s}$ denote peak-amplitude
coefficients, then, since both $\cos^2 n\theta$ and $\sin^2 n\theta$ average to $1/2$ over one fundamental period, 
\[
Q_n=\frac12
\left(V_{n,c}I_{n,s}-V_{n,s}I_{n,c}\right),
\]
Hence $Q_1=\tfrac12\left(1\cdot\epsilon-0\cdot G\right)=\frac{\epsilon}{2}$, 
$Q_2=\tfrac12\left(0\cdot G-1\cdot\frac{\epsilon}{2}\right)=-\frac{\epsilon}{4}$.
The DC component has no quadrature reactive-power term, while the generated
third- and fourth-harmonic currents do not contribute to diagonal $Q_n$
terms because the voltage has no corresponding harmonic components.
Therefore, the equal- and derivative-weighted scalar readouts give
\begin{equation}
 \QB=Q_1+Q_2=\frac{\epsilon}{4},
 \qquad
 \QD=Q_1+2Q_2=0.
 \label{eq:two-witness}
\end{equation}

The conflict is between two global calibration requirements. The derivative-weighted readout $Q_D$, calibrated to the reference fundamental, is internally consistent with the resistor-nulling criterion, whereas the equal-weighted readout $Q_B$ is consistent with preserving conventional reactive-power calibration on each harmonic plane.  Calling the second harmonic a new fundamental in a separate experiment rescales the derivative readout; it does not produce one fixed bilinear functional satisfying both requirements on the original space.

\section{Measurement Information Lost by Scalarization}
\label{sec:information}
The following result uses elementary rank--nullity to quantify the current information retained when waveform measurements are reduced to selected scalar power readouts. The contribution is this characterization of measurement-information loss within the stated model, not new linear algebra.

\subsection{Compatible currents at fixed voltage}
To quantify how much information these scalar measurements retain about an unknown current $i$ for a known voltage $v\ne0$, represent both waveforms in a finite-dimensional real inner-product space $H$ with $\dim H=d$. Suppose the available reduced data are the active power $P$ and
\begin{equation}
\begin{aligned}
q_j&=\ip{A_jv}{i},\qquad A_j^*=-A_j,\\
L_v&=\Span\{A_1v,\ldots,A_kv\},\qquad r_v=\dim L_v.
\end{aligned}
\label{eq:observations}
\end{equation}
Fig.~\ref{fig:flow} summarizes this distinction:
the complete compatible pair $(P,B)$ preserves the current at fixed
$v\neq0$, whereas the reduced scalar measurements observe only the
directions spanned by $A_jv$ in addition to the active direction.
The operators $A_j$ are specified independently of $i$. No common-shift assumption is needed for the next theorem, and DC components and arbitrary finite waveform bases are allowed.

\begin{figure}
\centering
\input{figures/information_flow.tex}
\caption{Fixed nonzero voltage: full compatible power data recover the current, while scalar readouts retain only the active direction and the span of $A_jv$; the remaining directions are unobserved.
}
\label{fig:flow}
\end{figure}
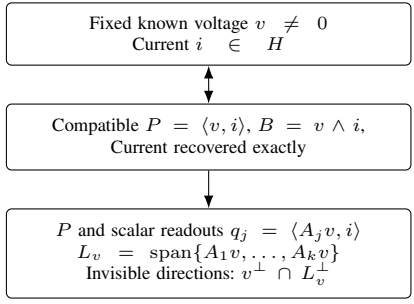

\begin{theorem}[Exact compatible-current dimension]
\label{thm:rank}
For compatible measured values in \eqref{eq:observations}, let
$\mathcal F_{v,P,q}\subset H$ denote the set of all currents $i$ that produce
the prescribed active power $P$ and scalar readouts $q_j$ for the known
voltage $v$. Then $\mathcal F_{v,P,q}$ is an affine space of dimension
\begin{equation}
 \dim\mathcal F_{v,P,q}=d-1-r_v.
 \label{eq:fiber-dim}
\end{equation}
If $i_*$ is any one compatible current, then
\begin{equation}
 \mathcal F_{v,P,q}=i_*+(v^\perp\cap L_v^\perp).
 \label{eq:fiber}
\end{equation}
Every current in $H$ is uniquely recoverable from these measurements if and only if $r_v=d-1$.
\end{theorem}

\begin{proof}
Since $P=\ip{v}{i}$, the component of $i$ parallel to $v$ is fixed, so every compatible current can be written uniquely as $i=(P/V^2)v+z$, with $z\in v^\perp$ and $V^2=\ip{v}{v}$. Since $A_j^*=-A_j$, we have $\ip{A_jv}{v}=0,$ i.e., $A_jv\in v^\perp$, and therefore $q_j=\ip{A_jv}{z}$. Thus, the readouts constrain $z$ only along subspace $L_v=\Span\{A_1v,\ldots,A_kv\}$ of dimension $r_v$.

If $i_*$ is one compatible current, then another current $i_*+\delta i$ gives the same $P$ and $q_j$ if and only if $\ip{v}{\delta i}=0$ and $\ip{A_jv}{\delta i}=0$ for all $j$, i.e., $\delta i\in v^\perp\cap L_v^\perp$. Hence, $\mathcal F_{v,P,q}=i_*+(v^\perp\cap L_v^\perp)$. Since $L_v\subseteq v^\perp$, $\dim v^\perp=d-1$, and $\dim L_v=r_v$, it follows that $\dim\mathcal F_{v,P,q}=d-1-r_v$. Therefore, the current is uniquely recoverable if and only if this dimension is zero, equivalently, if and only if $r_v=d-1$.
\end{proof}

Fig.~\ref{fig:fiber} illustrates the case $d=3$ and $r_v=1$: $P$ fixes the component of $i$ along $v$, the readout fixes the component along $L_v$, and the remaining direction $v^\perp\cap L_v^\perp$ is the compatible-current fiber, whose general dimension is $d-1-r_v$. Any nonzero perturbation invisible to these scalar measurements changes the complete voltage--current bivector $v\wedge i$, because $v\wedge\delta i=0$ implies $\delta i\parallel v$, while invisibility to $P$ requires $\delta i\perp v$. Thus, \emph{the scalar measurements can discard current information retained by the bivector}.

The relevant rank is evaluated at the actual voltage $v$. Distinct or linearly independent operators $A_j$ need not produce linearly independent vectors $A_jv$. Therefore, unrestricted recovery of $i$ requires at least $d-1$ scalar skew readouts whose vectors $A_jv$ span $v^\perp$, in addition to the known $v$ and $P$. This bound applies to linear scalar observations at fixed voltage; it is not a lower bound for arbitrary nonlinear encodings or for currents already constrained by a known model. With inconsistent or noisy data, the compatible-current set
$\mathcal F_{v,P,q}$ may be empty; in that case,
the dimension formula \eqref{eq:fiber-dim} does not apply.

Theorem~\ref{thm:rank} counts ambiguity in admissible current directions, rather than in arbitrary bivector coordinates. Each current $i\in\mathcal F_{v,P,q}$ defines a voltage--current pair $(v,i)$ and hence a corresponding bivector $v\wedge i$. The theorem does not assert that every such current can arise from a specified passive, causal, or reciprocal circuit family. Intersecting the affine fiber with the response set of such a circuit family can reduce the ambiguity. Thus, the theorem concerns information recoverable from the stated measurements, rather than realizability by a particular network model.

\begin{figure}[t]
\centering
\input{figures/current_fiber.tex}
\caption{
One nonzero scalar readout in a three-dimensional waveform space, shown in $v^\perp$ (the plane). $P$ fixes the component along $v$, while $q=\langle Av,z\rangle$ fixes the component of $z=i-Pv/V^2$ along the direction $Av$ (horizontal); the orthogonal (vertical) direction remains unobserved.}
\label{fig:fiber}
\end{figure}
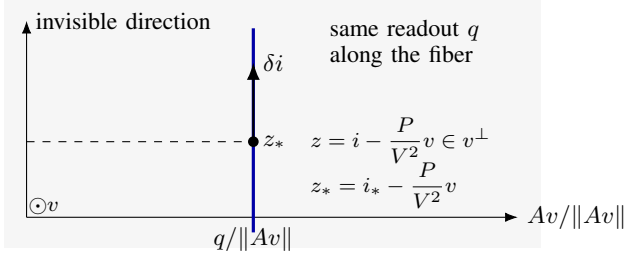

\subsection{Residual ambiguity under all invariant scalar readouts}
\begin{corollary}[Harmonic invisible directions]
\label{cor:harmonic}
Let $H$ be the sum of $N$ distinct positive-frequency scalar harmonic planes, with no DC, and let $v_n\ne0$ on every plane. Even all alternating bilinear common-shift-invariant scalar readouts, together with $P$ and $v$, leave exactly $N-1$ current directions invisible. They are
\begin{equation}
 \delta i=\sum_n c_nv_n,\qquad
 \sum_n c_nV_n^2=0,\qquad V_n=\norm{v_n}.
 \label{eq:invisible}
\end{equation}
\end{corollary}
\begin{proof}
The character argument of Proposition~\ref{prop:shifts} applies to any distinct positive orders. All these readouts together give precisely the separate $Q_n$. Their observation directions $J_nv_n$ are nonzero and mutually orthogonal, so $r_v=N$. Because $d=2N$, Theorem~\ref{thm:rank} gives $N-1$. Orthogonality to $J_nv_n$ forces each invisible component to be parallel to $v_n$, and orthogonality to $v$ gives the constraint in \eqref{eq:invisible}.
\end{proof}

These are in-phase harmonic redistributions at fixed total active power.
The scattered-current and conductance-dispersion interpretation is
established in CPC
\cite{CzarneckiSwietlicki1990,Czarnecki2008,Czarnecki2019}, with related geometric decompositions in \cite{Cieslinski2024}. The statement here identifies these established current directions as the kernel of the entire invariant scalar observation class; it does not introduce a new current component. If the $j$th scalar readout has the form
$q_j=\sum_{n=1}^N a_{jn}Q_n$, let $[a_{jn}]$ denote the corresponding
readout-weight matrix. When all voltage harmonic planes are excited,
$r_v=\rank[a_{jn}]$. Thus, even if the readouts recover every diagonal
harmonic quantity $Q_n$, they determine only the quadrature component of
the current in each harmonic plane. The in-phase components remain
constrained only through their aggregate contribution to $P$, leaving
$N-1$ independent redistribution directions. This $N-1$ count assumes that all $N$ voltage harmonic planes are excited.
If only $s$ planes are excited in the $2N$-dimensional positive-frequency
space, the full invariant readout class has rank $s$, and
\begin{equation}
\dim\mathcal F_{v,P,q}=2N-1-s.    
\end{equation}
The ambiguity consists of $s-1$ in-phase redistribution directions among
the excited planes and two unconstrained current coordinates on each of
the $N-s$ unexcited planes. If a DC component is included, its skew block
is zero, and the general formula $d-1-r_v$ remains applicable, with the
DC current constrained only through active power.

\subsection{Implications for metering and compensation}
For a chosen scalar measurement set, Theorem~\ref{thm:rank} identifies the invisible current perturbations exactly as $v^\perp\cap L_v^\perp$. Identical readings determine the same current only when this space is zero. Measurement sufficiency therefore depends on the selected readouts and the intended objective; the result characterizes deterministic information loss, rather than measurement uncertainty or general system observability.

The rank result has a direct metering interpretation: scalar power
measurements can determine the quadrature harmonic components while
discarding information about how in-phase current is distributed among
harmonics.
For an LTI harmonic record on excited voltage planes, write $i_n=g_nv_n+b_nJ_nv_n$. The separate powers $Q_n=b_nV_n^2$ determine the quadrature coefficients, while $P=\sum g_nV_n^2$ determines only one weighted sum of conductances. Perturbations $g_n\mapsto g_n+c_n$ satisfying \eqref{eq:invisible} are invisible. Complete phasors directly contain these conductances. The information is discarded only when measurements are reduced to the selected scalars.

This ambiguity also has a direct compensation interpretation. Currents with identical total active power and all time-shift-invariant alternating scalar readouts can still differ in RMS current, harmonic conductance distribution, conductor $I^2R$ burden, compensation requirements, and waveform shape because invisible in-phase redistribution changes the harmonic conductance pattern, as illustrated by \eqref{eq:pair}--\eqref{eq:pair-difference}. Thus, a scalar meter may be sufficient for one measurement objective but not another. The rank result identifies the missing current directions; a known constitutive model, harmonic active-power measurements, or complete current coefficients can supply additional information. It does not, however, prescribe objective-specific measurement requirements or a particular compensator, control strategy, or tariff.

\section{Physical Interpretation and Applicability Boundaries}
\label{sec:physical}
\subsection{Storage-imbalance interpretation in the LTI setting}
For a passive lumped LTI RLC realization in a periodic steady state, with constant ideal transformers permitted, all ports accounted for, and no active, gyrator, or radiative elements in the asserted realization class, the classical branch balance gives
\begin{equation}
 Q_n=2\omega_n(\overline W_{L,n}-\overline W_{C,n}),
 \qquad \omega_n=n\omega_0>0.
 \label{eq:lti-energy}
\end{equation}
Here $\overline W_{L,n}$ is the cycle-mean magnetic energy of the $n$th harmonic, summed over all inductors in the specified realization, and $\overline W_{C,n}$ is the corresponding cycle-mean electric energy summed over all capacitors. One checks the branch factors directly: $Q_{L,n}=\omega_nLI_n^2$, $\overline W_{L,n}=LI_n^2/2$, where $I_n$ is the branch RMS current; for a capacitor, $Q_{C,n}=-\omega_n C V_n^2$ and $\overline W_{C,n}=C V_n^2/2$, with branch RMS voltage $V_n$. Kirchhoff balance cancels internal port powers; ideal transformers contribute no storage. This energy interpretation uses a specified realization in addition to terminal waveforms, consistent with classical nonlinear and differentiated-power circuit theory \cite{Jeltsema2003,GarciaCanseco2005}. Equation~\eqref{eq:lti-energy} represents magnetic-minus-electric storage imbalance, not total stored energy. In particular,
\(
Q_n/(2\omega_n)=\overline W_{L,n}-\overline W_{C,n},
\)
so a small or even zero value can result from cancellation between two individually nonzero stored-energy terms.
The relation in \eqref{eq:lti-energy} also excludes energy stored in any DC component. Using the derivative weighting $Q_D=\sum_n nQ_n$ and $\omega_n=n\omega_0$ gives
\begin{equation}
 \QD=\frac{2}{\omega_0}\sum_n\omega_n^2
 \left(\overline W_{L,n}-\overline W_{C,n}\right),
 \label{eq:weighted-energy}
\end{equation}
within the same finite-dimensional LTI setting. 
Thus, the derivative readout $Q_D$ and the inverse-frequency readout
$e_\Delta$ in Table~\ref{tab:selectors} summarize the same harmonic
magnetic--electric energy imbalance with different frequency weightings:
$Q_D$ emphasizes higher-frequency harmonics, whereas $e_\Delta$ equals the summed modal energy imbalance under the stated LTI
assumptions.

\vspace{-0.1in}
\subsection{Nonlinear and time-varying circuit elements}
Consider smooth time-invariant ideal laws $q=q(v)$ and $\phi=\phi(i)$ on intervals containing zero and the respective trajectory ranges, with $q(0)=\phi(0)=0$ and $q',\phi'>0$. Set physical energy to zero at the origin. For the capacitor, $W_C(v)=\int_0^v\xi q'(\xi)\,\dd\xi$ and current $i=\dot q$; for the inductor, $W_L(i)=\int_0^i\xi\phi'(\xi)\,\dd\xi$ and voltage $v=\dot\phi$. On smooth periodic trajectories, substitution in \eqref{eq:derivative} gives
\begin{equation}
 Q_{D,C}=-\frac{\avg{q'(v)\dot v^2}}{\omega_0},
 \qquad
 Q_{D,L}=\frac{\avg{\phi'(i)\dot i^2}}{\omega_0}.
 \label{eq:nonlinear-motion}
\end{equation}
These differentiated-power relations depend jointly on the device
constitutive law and on waveform motion
\cite{Jeltsema2003,GarciaCanseco2005}. In charge or flux coordinates,
they can be expressed through second derivatives of the corresponding
energy or coenergy functions, but they do not generally reduce to a
constant multiple of mean stored energy. Moreover, energy and coenergy
coincide only for linear constitutive laws. A constant biased state may
therefore store nonzero energy while giving zero derivative readout.

The time-invariance assumption in the resistor-nulling result is also
essential \cite{Jeltsema2014}. For a positive time-varying conductance
$i=g(t)v$, periodic integration by parts gives
\begin{equation}
 \QD=\frac{\avg{\dot g\,v^2}}{2\omega_0}.
 \label{eq:modulated}
\end{equation}
For example, with $v=a\cos\theta$ and
$g=g_0+g_2\sin2\theta$, $g_0>|g_2|$, one obtains
$\QD=a^2g_2/4$. Thus $\QD$ can be positive or negative even though
$vi=gv^2\ge0$ at all times. The ideal one-port itself has no storage
state; any physical mechanism that modulates $g(t)$ would require
additional ports or internal states to be modeled explicitly. Hence,
neither a nonzero $\QD$ nor its sign is, by itself, a universal indicator
of electrical energy storage.

Finally, recovering the terminal current is not the same as recovering
internal stored energy. Dissipativity theory allows multiple storage
functions to certify the same terminal behavior
\cite{Willems1972I,Willems1972II}, and such a storage function need not
coincide with the physical energy of a specified device. Likewise,
Proposition~\ref{prop:reconstruction} reconstructs the terminal current
waveform, but not the internal state, constitutive law, or circuit
topology. Any stronger statement about internal energy therefore requires
additional assumptions on the physical realization and its admissible
states.

\section{Exact Examples and Numerical Illustration}
\label{sec:examples}
\subsection{A globally passive cubic resistor}
The two-harmonic example in Section~\ref{sec:conflict} uses a smooth,
globally passive resistor law that reproduces a quadratic nonlinearity
over the voltage range of interest. A simpler globally passive polynomial
example is
\begin{equation}
 v=a\cos\theta+b\sin3\theta,\quad
 i=Gv+\alpha v^3,\quad G,\alpha>0.
 \label{eq:cubic}
\end{equation}
Here $f'(v)=G+3\alpha v^2>0$ and $vi=Gv^2+\alpha v^4\ge0$. The element has no dynamic state. In $v^3$, the sine coefficient at harmonic one is $3a^2b/4$, and the cosine coefficient at harmonic three is $a^3/4$. Therefore
\[
 Q_1=\frac{3\alpha a^3b}{8},\quad
 Q_3=-\frac{\alpha a^3b}{8},\quad
 \QB=\frac{\alpha a^3b}{4},\quad \QD=0. 
\]
The current also contains fifth-, seventh-, and ninth-harmonic components, whose diagonal powers vanish because the corresponding voltage harmonics are absent, although their contributions to the full voltage--current bivector need not. With $E_n=e_{n,c}\wedge e_{n,s}$ denoting the oriented unit bivector of the $n$th harmonic plane, the diagonal bivector is proportional to $3E_1-E_3\ne0$, while $Q_D=Q_1+3Q_3=0$. Thus, cancellation of a scalar projection does not imply cancellation of the underlying bivector components.

Table~\ref{tab:resistor} compares the analytic values with Fourier readouts for $a=2$~V, $b=1$~V, $G=0.4$~S, and $\alpha=0.2$~A/V$^3$. The current coefficients were calculated on 16\,384 equally spaced samples over one period, with all polynomial harmonics retained. Independent 768-node Gauss--Legendre quadrature of $-\dot v i/\omega_0$ gives $1.27898\times10^{-13}$~VA for the theoretically zero derivative readout. These calculations illustrate the exact identities; the universal cancellation follows from the proof.

\begin{table}[t]
\caption{Cubic-resistor example: analytic and Fourier readouts}
\label{tab:resistor}
\centering\footnotesize
\begin{tabular}{@{}lrr@{}}
\toprule
Readout & Exact value (VA) & Fourier value (VA) \\
\midrule
$Q_1$ & $0.6$ & $0.600000000000$ \\
$Q_3$ & $-0.2$ & $-0.200000000000$ \\
$\QB=Q_1+Q_3$ & $0.4$ & $0.400000000000$ \\
$\QD$ & $0$ & $-2.53484\times10^{-17}$ \\
\bottomrule
\end{tabular}
\end{table}

\subsection{Identical scalars, readouts, different current distributions}
Use the RMS basis $(e_{1,c},e_{1,s},e_{2,c},e_{2,s})$ and normalized units. For $G>1$, take $v=(1,0,1,0)$ and
\begin{equation}
 i_0=(G,0,G,0),\qquad i_1=(G+1,0,G-1,0).
 \label{eq:pair}
\end{equation}
Both records have $P=2G$ and $Q_1=Q_2=0$, so every readout in
\eqref{eq:weights} agrees. Nevertheless,
\begin{equation}
 v\wedge i_0=0,\quad
 v\wedge i_1=-2e_{1,c}\wedge e_{2,c},
 \quad
 \norm{i_1}^2-\norm{i_0}^2=2.
 \label{eq:pair-difference}
\end{equation}
Thus, a meter retaining $P$ and any collection of the invariant reactive-power readouts in \eqref{eq:weights} would report identical power quantities for these two records, although they differ along the single invisible in-phase direction for $N=2$ identified by Theorem~\ref{thm:rank}. The harmonic conductances change from $(G,G)$ to $(G+1,G-1)$, with $I_0=\sqrt{2G^2}$ and $I_1=\sqrt{2G^2+2}$; for a frequency-independent series resistance $R_\ell>0$, the corresponding loss difference is $R_\ell(I_1^2-I_0^2)=2R_\ell$. At the same $v$ and $P$, the minimum-RMS active current is $Pv/V^2=i_0$, so an ideal shunt compensator attaining that objective would cancel $i_1-i_0$ in the second record but nothing in the first. Both records have positive harmonic conductance data, but the example concerns admitted waveform records rather than synthesis of arbitrary passive causal devices. Complete harmonic current phasors, or an additional RMS-current measurement, distinguish the pair; such measurements enlarge the observation set beyond the scalar readouts considered in Theorem~\ref{thm:rank}.

\vspace{-0.1in}
\section{Discussion and Conclusions}
\label{sec:discussion}

The sinusoidal case is exceptional because a single oriented harmonic plane has only one quadrature direction. For known nonzero voltage, the active and quadrature directions span that plane, so $P$ and one calibrated $Q$ recover the current. With multiple harmonic planes, total active power compresses all in-phase components into one scalar, while time-shift-invariant readouts retain only selected quadrature information. The resulting ambiguity is characterized by \eqref{eq:fiber}; for the full invariant class with all voltage planes excited, it reduces to the $N-1$ in-phase redistribution directions in \eqref{eq:invisible}.

Different scalar readouts therefore reflect different structural requirements and measurement objectives. Equal harmonic weighting preserves conventional plane-by-plane calibration, derivative weighting is selected by the stated memoryless-resistor nulling criterion and fixed temporal reference, and inverse-frequency weighting represents a storage imbalance under the restricted LTI assumptions. These properties are not equivalent: resistor nulling does not imply absence of all nonactive current, and no single scalar generally determines a compensation objective or physical stored energy. More generally, active power together with skew scalar readouts leaves a compatible-current fiber of dimension $d-1-r_v$; additional scalar formulas provide new information only when their associated directions $A_jv$ increase the observation rank. Complete waveform or phasor measurements retain the missing current information, while a known circuit model may further restrict the admissible current space.

The exterior formulation makes the retained and discarded directions explicit but does not add information beyond a complete voltage--current description. Its role is to expose the geometry of scalarization and its information limits. The results apply to the stated periodic, finite-dimensional, fixed-reference setting; multiphase multiplicities, hysteresis, explicit modulation, sparse harmonic support, nonlinear measurements, and model-restricted synthesis require additional assumptions. Within this scope, the converse selection result explains why universal smooth monotone memoryless-resistor nulling forces the known derivative/Iliovici weighting, while the rank theorem identifies exactly which current directions remain invisible after scalarization. Together, these results clarify what nonsinusoidal scalar power measurements specify and whether the retained information is sufficient for a given measurement objective.

\vspace{-0.1in}
\section*{Acknowledgment}
The author used OpenAI ChatGPT to assist in checking mathematical derivations and proofs and in editing Sections II--VIII. All mathematical developments and conclusions were independently reviewed and verified by the author.

\bibliographystyle{IEEEtran}
\bibliography{references}
\end{document}

%% file: figures/information_flow.tex
% Schematic from reconstruction and fixed-voltage rank theorem; no data fit.
\begin{tikzpicture}[
 box/.style={draw,rounded corners=2pt,align=center,text width=5.0cm,
 inner sep=5pt,font=\scriptsize},
 >=Latex,node distance=5mm]

\node[box] (current)
{Fixed known voltage $v\ne0$\\Current $i\in H$};

\node[box,below=of current] (full)
{Compatible $P=\langle v,i\rangle$, $B=v\wedge i$,\\
Current recovered exactly};

\node[box,below=of full] (scalars)
{$P$ and scalar readouts $q_j=\langle A_jv,i\rangle$\\
$L_v=\operatorname{span}\{A_1v,\ldots,A_kv\}$\\
Invisible directions: $v^\perp\cap L_v^\perp$};

\draw[<->] (current)--(full);
\draw[->] (full)--(scalars);

\end{tikzpicture}

%% file: figures/current_fiber.tex
% Generic d=3, r_v=1 schematic in the active-power residual plane.
\begin{tikzpicture}[x=1cm,y=1cm,>=Latex,font=\small]
\fill[gray!7] (-.3,-.4) rectangle (6.8,2.9);

% Label the displayed plane
%\node[anchor=north west,font=\footnotesize] at (.55,2) {$v^\perp$};

% Voltage direction normal to the displayed plane
\node at (.15,.15) {$\odot$};
\node[font=\footnotesize] at (.35,.15) {$v$};

% Axes in v^\perp
\draw[->] (0,0)--(6.5,0)
  node[right] {$Av/\norm{Av}$};
\draw[->] (0,0)--(0,2.6)
  node[right] {invisible direction};

% Compatible-current fiber in the residual plane
\draw[very thick,blue!65!black] (3,-.2)--(3,2.5);

% One compatible residual current z_*
\draw[dashed] (0,1)--(3,1);
\fill (3,1) circle (2pt) node[right] {$z_*$};

% Invisible perturbation
\draw[->,thick] (3,1)--(3,2.05)
  node[right] {$\delta i$};

% Measurement interpretation
\node[align=left] at (5.0,2.3)
  {same readout $q$\\along the fiber};

\node[below] at (3,0)
  {$q/\norm{Av}$};

\node[align=left,font=\footnotesize] at (4.95,.72)
  {$z=i-\dfrac{P}{V^2}v\in v^\perp$\\
   $z_*=i_*-\dfrac{P}{V^2}v$};

\end{tikzpicture}

%% file: references.bib
@article{KosobudzkiLadniak2026,
  author={Kosobudzki, Grzegorz and {\L}adniak, Leszek},
  title={Generalization of Reactive Power Definition for Periodical Waveforms},
  journal={International Journal of Electrical and Computer Engineering},
  year={2026}, volume={16}, number={1}, pages={102--110},
  doi={10.11591/ijece.v16i1.pp102-110}
}

@article{HongDeLeon2015,
  author={Hong, Tianqi and de Le{\'o}n, Francisco},
  title={Lissajous Curve Methods for the Identification of Nonlinear
  Circuits: Calculation of a Physical Consistent Reactive Power},
  journal={IEEE Transactions on Circuits and Systems I: Regular Papers},
  year={2015},
  volume={62},
  number={12},
  pages={2874--2885},
  doi={10.1109/TCSI.2015.2495780}
}

@article{LaWhiteIlic1997,
  author={LaWhite, Niels and Ili{\'c}, Marija D.},
  title={Vector Space Decomposition of Reactive Power for Periodic Nonsinusoidal Signals},
  journal={IEEE Transactions on Circuits and Systems I: Fundamental Theory and Applications},
  year={1997}, volume={44}, number={4}, pages={338--346},
  doi={10.1109/81.563623}
}

@article{Czarnecki1987,
  author={Czarnecki, Leszek S.},
  title={What Is Wrong with the {Budeanu} Concept of Reactive and Distortion Power and Why It Should Be Abandoned},
  journal={IEEE Transactions on Instrumentation and Measurement},
  year={1987}, volume={IM-36}, number={3}, pages={834--837}
}

@article{CzarneckiSwietlicki1990,
  author={Czarnecki, Leszek S. and Swietlicki, Tadeusz},
  title={Powers in Nonsinusoidal Networks: Their Interpretation, Analysis, and Measurement},
  journal={IEEE Transactions on Instrumentation and Measurement},
  year={1990}, volume={39}, number={2}, pages={340--345}
}

@article{Czarnecki2008,
  author={Czarnecki, Leszek S.},
  title={Currents' Physical Components ({CPC}) Concept: A Fundamental of Power Theory},
  journal={Przegl{\k a}d Elektrotechniczny},
  year={2008}, volume={84}, number={6}, pages={28--37}
}

@article{Czarnecki2019,
  author={Czarnecki, Leszek S.},
  title={Currents' Physical Components ({CPC})--Based Power Theory: A Review. Part {I}: Power Properties of Electrical Circuits and Systems},
  journal={Przegl{\k a}d Elektrotechniczny},
  year={2019}, volume={95}, number={10}, pages={1--11},
  doi={10.15199/48.2019.10.01}
}

@article{Willems2011,
  author={Willems, Jacques L.},
  title={{Budeanu}'s Reactive Power and Related Concepts Revisited},
  journal={IEEE Transactions on Instrumentation and Measurement},
  year={2011}, volume={60}, number={4}, pages={1182--1186},
  doi={10.1109/TIM.2010.2090704}
}

@misc{Jeltsema2014,
  author={Jeltsema, Dimitri and van der Woude, Jacob W. and Hartman, Marek T.},
  title={A Novel Time-Domain Perspective of the {CPC} Power Theory: Single-Phase Systems},
  year={2014}, howpublished={arXiv:1403.7842v2}
}

@article{Jeltsema2003,
  author={Jeltsema, Dimitri and Ortega, Romeo and Scherpen, Jacquelien M. A.},
  title={On Passivity and Power-Balance Inequalities of Nonlinear {RLC} Circuits},
  journal={IEEE Transactions on Circuits and Systems I: Fundamental Theory and Applications},
  year={2003}, volume={50}, number={9}, pages={1174--1179},
  doi={10.1109/TCSI.2003.816332}
}

@inproceedings{GarciaCanseco2005,
  author={Garc{\'i}a-Canseco, Elo{\'i}sa and Gri{\~n}{\'o}, Robert and Ortega, Romeo and Salichs, Miguel and Stankovic, Alexander},
  title={Power Factor Compensation of Electrical Circuits: A Control Theory Viewpoint},
  booktitle={Congreso Nacional de Control Autom{\'a}tico},
  address={Cuernavaca, Mexico}, year={2005}, month=oct,
  note={Paper AMCA05031}
}

@article{SupertiFurgaPinola1994,
  author={{Superti Furga}, G. and Pinola, L.},
  title={The Mean Generalized Content: A Conservative Quantity in Periodically-Forced Non-Linear Networks},
  journal={European Transactions on Electrical Power},
  year={1994}, volume={4}, number={3}, pages={205--212},
  doi={10.1002/etep.4450040305}
}

@article{SupertiFurga1994,
  author={{Superti Furga}, G.},
  title={Searching for a Generalization of the Reactive Power---A Proposal},
  journal={European Transactions on Electrical Power},
  year={1994}, volume={4}, number={5}, pages={411--417},
  doi={10.1002/etep.4450040515}
}

@article{Menti2007,
  author={Menti, Anthoula and Zacharias, Thomas and Milias-Argitis, John},
  title={Geometric Algebra: A Powerful Tool for Representing Power under Nonsinusoidal Conditions},
  journal={IEEE Transactions on Circuits and Systems I: Regular Papers},
  year={2007}, volume={54}, number={3}, pages={601--609},
  doi={10.1109/TCSI.2006.887608}
}

@misc{Montoya2020,
  author={Montoya, F. G. and Rold{\'a}n-P{\'e}rez, J. and Alcayde, A. and Arrabal-Campos, F. M. and Ba{\~n}os, R.},
  title={Geometric Algebra Power Theory in Time Domain},
  year={2020}, howpublished={arXiv:2002.05458v4}
}

@article{Montoya2021,
  author={Montoya, Francisco G. and Ba{\~n}os, Ra{\'u}l and Alcayde, Alfredo and Arrabal-Campos, Francisco Manuel and Rold{\'a}n-P{\'e}rez, Javier},
  title={Vector Geometric Algebra in Power Systems: An Updated Formulation of Apparent Power under Non-Sinusoidal Conditions},
  journal={Mathematics}, year={2021}, volume={9}, number={11}, pages={1295},
  doi={10.3390/math9111295}
}

@article{Cieslinski2024,
  author={Cie{\'s}li{\'n}ski, Jan L. and Walczyk, Cezary J.},
  title={Geometric Algebra Framework Applied to Single-Phase Linear Circuits with Nonsinusoidal Voltages and Currents},
  journal={Electronics}, year={2024}, volume={13}, number={19}, pages={3926},
  doi={10.3390/electronics13193926}
}

@article{Willems1972I,
  author={Willems, Jan C.},
  title={Dissipative Dynamical Systems Part {I}: General Theory},
  journal={Archive for Rational Mechanics and Analysis},
  year={1972}, volume={45}, pages={321--351},
  doi={10.1007/BF00276493}
}

@article{Willems1972II,
  author={Willems, Jan C.},
  title={Dissipative Dynamical Systems Part {II}: Linear Systems with Quadratic Supply Rates},
  journal={Archive for Rational Mechanics and Analysis},
  year={1972}, volume={45}, pages={352--393},
  doi={10.1007/BF00276494}
}

@standard{IEEE1459_2025,
  title        = {{IEEE Standard Definitions for the Measurement of Electric Power Quantities Under Sinusoidal, Nonsinusoidal, Balanced, or Unbalanced Conditions}},
  organization = {IEEE},
  number       = {IEEE Std 1459-2025},
  year         = {2025},
  month        = may
}
